\documentclass[a4paper,11pt]{article}
\usepackage{amsthm,amsmath}
\usepackage{amssymb}
\usepackage{amsfonts}
\usepackage{color}
\usepackage[english]{babel}

\usepackage{mathrsfs}

\newcommand{\cmU}{\mathcal{U}}
\newcommand{\spin}{\sf{s}}
\newcommand{\ra}{\rangle}
\newcommand{\la}{\langle}

\newcommand{\I}{\mathbb{I}}
\newcommand{\Sp}{\mathbb{S}}

\usepackage{hyperref}
\numberwithin{equation}{section}
\newtheorem{theorem}{Theorem}[section]
\newtheorem{Def}[theorem]{Definition}
\newtheorem{Pro}[theorem]{Proposition}
\newtheorem{lem}[theorem]{Lemma}
\newtheorem{remark}[theorem]{Remark}

\newtheorem{Cor}[theorem]{Corollary}

\newcommand{\eps}{\varepsilon}

\newcommand{\n}{\mathbb{N}}
\newcommand{\R}{\mathbb{R}}
\newcommand{\cc}{\mathbb{C}}

\newcommand{\di}{\displaystyle}
\def\beq{\begin{equation}}   \def\eeq{\end{equation}}
\def\bea{\begin{eqnarray}}  \def\eea{\end{eqnarray}}

\begin{document}

\title{Spin-orbit interaction with large spin \\in the semi-classical regime}
\date{}
\author{
Didier Robert\footnote{ Laboratoire de Mathématiques Jean Leray, Université de Nantes,
	2 rue de la Houssini\`ere, BP 92208, 44322 Nantes Cedex 3, France\newline
\textit{Email: }{didier.robert@univ-nantes.fr}}
}
\maketitle

\begin{abstract}
We consider  the time dependent Schr\"odinger equation with a coupling spin-orbit in the semi-classical regime $\hbar\searrow 0$  and large spin number $\spin\rightarrow +\infty$
such that  $\hbar^\delta\spin=c$ where $c>0$ and $\delta>0$ are constant.  The initial state $\Psi(0)$  is a product of an orbital coherent state in $L^2(\R^d)$  
and a spin coherent state in a spin irreducible  representation space ${\mathcal H}_{2\spin +1}$. For $\delta <1$, at the leading order in $\hbar$,  
the   time evolution  $\Psi(t)$ of $ \Psi(0)$  is well approximated by the product of an orbital and a spin coherent state. Nevertheless for $1/2<\delta<1$ the quantum orbital leaves the classical orbital. For  $\delta=1$ we prove that this last claim is no more true when  the interaction depends on the orbital variables. 
   For the Dicke model,   we prove that the orbital partial trace of the projector  on $\Psi(t)$ is a  mixed state in $L^2(\R)$  
  for small $t>0$.
   	 \end{abstract}

 \section{Introduction}
 We consider here the Schr\"odinger equation for a system of particles with large  spin number when the spin and the position variables are   coupled: 
\beq\label{eq:sch0}
\begin{aligned}
 i\hbar\partial_t	\Psi(t) &= \hat H(t)\Psi(t), \\
  \Psi(0) &=\varphi_{z_0}\otimes\psi_{\bf n_0},\;\;{\rm in}\;\; L^2(\R^d)\otimes{\mathcal H}_{2\spin +1} 
  \end{aligned}
  \eeq
 where $\varphi_{z_0}, \psi_{\bf n_0}$ are  respectively Schr\"odinger,  spin coherent states and 
\beq\label{H:spin}
  \hat H(t) = \hat H_0(t) +\hbar\widehat{\mathfrak{C}}(t)\cdot{\bf S}.
  \eeq
  $\hat H_0(t)$ has  a scalar $\hbar$-Weyl symbol as well  as $\hat{\mathfrak C}_j(t)$ and 
   ${\bf S}=(S_1, S_2, S_3)$ are  spin matrices representation in an  irreducible space   ${\mathcal H}_{2\spin+1}$ of dimension $2\spin+1$, where $\spin$ is the total spin number.\\
   A particular case is the Pauli equation for the electron ($\spin =1/2$).
   $$
   \hat H = \frac{1}{2}\left(\hbar D -A\right)^2 +\hat V +\hbar B\cdot \sigma,
   $$
   where $D=i^{-1}\nabla_x$, $x\in\R^3$, $A=(A_1,A_2,A_3)$ is a magnetic potential, $V$ an electric potential,   $B=(B_1,B_1,B_1)$ the magnetic field in $\R^3$; $\sigma=(\sigma_1,\sigma_2,\sigma_3)$
   are the Pauli matrices:
   $$
   \sigma_1=\begin{pmatrix}0 & 1\\1 & 0\end{pmatrix},\;\sigma_2=\begin{pmatrix}0 & -i\\i & 0\end{pmatrix},\;\sigma_3=\begin{pmatrix}1 & 0\\0 & -1\end{pmatrix}.
   $$
   For $\spin =1/2$ we have $S_k = \frac{\sigma_k}{2}$.\\
  More generally the Weyl symbol of the interaction  is $\hbar H_1(t,X)$ where\\
 $\di{H_1(t, X)= \sum_{1\leq k\leq 3}\mathfrak{C}_k(t, X)S_k}=\mathfrak{C}(t,X)\cdot{\bf S}$.  \\ 
 If the spin $\spin$  is fixed, in the semi-classical regime, $\hbar\searrow 0$,  
 $H_1(t,X)$ is the subprincipal symbol of $\hat H(t)$. Hence we can get a semiclassical approximation at any order for $\Psi(t)$  using generalized  coherent states as it will be recalled later. For details see \cite{CR}, Chapter 14 and \cite{FLR} for matrix principal symbol $H_0(t)$. \\
  The spin matrices $S_k$ are realized as  hermitian matrices  in the Hermitian  space  ${\mathcal H}_{2\spin +1}$ such that 
 $ S_k = d{\mathcal D}^{(\spin)}(\sigma_k/2)$ are defined by the derivative at $\I$ of the  irreducible  representation ${\mathcal D}^{(\spin)}$ in  ${\mathcal H}_{2\spin +1}$.      
 In particular from the Lie algebra $\mathfrak{su}(2)$ we have the commutation relations
 \beq
 [S_k, S_\ell] =i\epsilon_{k,\ell,m}S_m.
  \eeq
  $\mathfrak{C}_k(t,X)$ are real scalar   symbols in $X\in\R^{2d}$. So the full  Weyl symbol of $\hat H(t)$ is the matrix
 $H(t,X) = H_0(t,X)\I +\hbar \mathfrak{C} (t,X)\cdot{\bf S}$, with the spin operator ${\bf S}=(S_1, S_2, S_3)$.\\
 For simplicity we shall assume in all this paper that the Weyl symbols $H_0(t)$ and $\mathfrak{C} (t,X)$ are subquadratic  \cite{CR}(p.409). It follows that the quantum Hamiltonian $\hat H(t)$ generates a propagator $\cmU(t,t_0)$  in the Hilbert space $L^2(\R^d, {\mathcal H}_{2\spin +1})$. (see Proposition 123 of \cite{CR} easily extended to systems).
 
% {\bf An example: the Dicke model}.  This model is considered in quantum optics because it describes the interaction between light and matter  when the light (photons) is  a field operator 
% like quantum harmonic oscillators  and atoms of matter being in two spin levels. \textcolor{red}{ To detail.....}

 %%%%%%%%%%%%%%%%%%%%%%%%%%%%%%%%%%%%%%%%%%%%%%%%%
 %%%%%%%%%%%%%%%%%%%%%%%%%%%%%%%%%%%%%%%%%%%%%
It is known (see \cite{CR} ( Ch.1 and Chap.7)  that the Schr\"odinger coherent states are labelled by the phase space $T^*(\R^d)$ and the spin (or atomic) coherent states are labelled by the sphere $\Sp^2$. 
 So it is natural to study the semi-classical limit of the propagator  with  the phase space $T^*(\R^d)\times{\Sp}^2$. In particular we can reformulate the propagation for  coherent states $\varphi_z\otimes \psi_{\bf n}$ labelled by $(z,{\bf n})\in\R^{2d}\times\Sp^2$  where $\varphi_z$ is a  Schr\"odinger coherent state  and  $\psi^{(\spin)}_{\bf n}= \psi_{\bf n}$ is a spin coherent state.\\
 Let us recall our notations. 
 \begin{itemize}
 \item \underline{Heisenberg translations in $L^2(\R^d)$}:
 $$
 \hat T(z) = \exp\left(\frac{i}{\hbar}p\cdot\hat x -q\cdot\hat p\right)
 $$
 where $z=(q,p)\in\R^d\times\R^d$, $\hat x$ is mutiplication by $x$ and $\hat p=\frac{\hbar}{i}\nabla_x$.
 \item \underline{Schr\"odinger coherent states}:
 $$
 \varphi_z = \hat T(z)\varphi_0,\;\; \varphi_0(x) = (\pi\hbar)^{-d/4}\exp\left(\frac{\vert x\vert^2}{2\hbar}\right)
 $$
 \item \underline{$SU(2)$ and the sphere $\Sp^2$}: let 
  $$
 \mathbf n=(\sin\theta\cos\varphi,\sin\theta\sin\varphi,\cos\theta),\; 0 \leq\theta <\pi,\; 0\leq\varphi<2\pi;
 $$
To any $\mathbf n\in\Sp^2$ we associate the transformation in $SU(2)$,
 \beq\label{expn}
 g=g_{\mathbf n}=\exp\left(i\frac{\theta}{2}((\sin\varphi)\sigma_1 -(\cos\varphi)\sigma_2)\right)
 \eeq
\item \underline{irreducible representations and spin coherent states}:
let $\spin$ be an half integer and ${\mathcal D}^{\spin}$ the irreducible  representation in the Hilbert space ${\mathcal H}_{2\spin +1}$ of dimension $2\spin +1$. To ${\mathbf n}\in\Sp^2$ is associated  the spin coherent states \footnote{this construction follows from the computation  of the isotropy subgroups of the $SU(2)$ action: we get roughly $SU(2)/{\rm ISO}\approx \Sp^2$ (see \cite{CR} for details).}
$$
\psi_{\mathbf n} = {\mathcal D}^{\spin}(g_{\mathbf n})\psi_0,
$$
where $\psi_0$ is a unit eigenvector of $S_3$  in ${\mathcal H}_{2\spin +1}$ with minimal eigenvalue $-\spin$.
 \end{itemize}
 
 Let us consider   uncorrelated initial states $\Psi_{z_0,{\bf n}_0} =\varphi_{z_0}\otimes\psi^{(\spin)}_{{\bf n}_0}$, where  $\varphi_{z_0}$ is a Schr\"odinger (orbital) coherent state  and 
  $\psi^{(\spin)}_{{\bf n}_0}$ a spin coherent state. 
 When the spin number $\spin$ is fixed we have the following result proved in \cite{BG} and revisited in \cite{CR}. 
  \begin{theorem}\cite{BG}\label{thm:fsp}
 For the initial state $\Psi_{z_0,{\bf n}_0} =\varphi^{\Gamma_0}_{z_0}\otimes\psi^{(\spin)}_{{\bf n}_0}$
 we have
\beq\label{propag_spin1}
  \cmU(t,t_0)\Psi_{z_0,{\bf n}_0}= {\rm e}^{\frac{i}{\hbar}S(t,t_0) +i\spin\alpha(t)}\varphi_{z(t)}^{\Gamma(t)}\otimes\psi^{(\spin)}_{{\bf n}(t)}
  + O(\sqrt\hbar)
\eeq
  where  $z_0\mapsto z(t)$ is the classical flow for the  Hamiltonian for  $H_0(t)$,  $S(t,t_0)$ is the classical action, the  covariance matrix $\Gamma(t)$ is computed from  the dynamics generated by the linearized flow 
  of  $H_0(t)$ and $\alpha(t)$ is a real phase computed from the spin motion ${\bf n}(t)$ which  satisfies the Landau-Lifshitz  \cite{LL} equation
 \beq\label{eq:LL}
 \partial_t{\bf n}(t) =\mathfrak{C}(t,z(t))\wedge{\bf n}(t).
 \eeq
 \end{theorem}
\begin{remark} In Theorem \ref{thm:fsp} the motion of the spin depends on the motion along the orbit  but the orbit motion is independent of the spin.
\end{remark}
\begin{remark}
Theorem \ref{thm:fsp} is a particular case of propagation of coherent states for  systems with  a subprincipal term of order $\hbar$ and a principal term $H_0(t)$ (here scalar)  without crossing eigenvalues.  Moreover one  can get  a complete asymptotic expansion  in power of $\sqrt\hbar$ mod $O(\hbar^\infty)$.\\
For smooth crossings  some non-adiabatic results are proved in \cite{FLR}. 
\end{remark} 
The  first  new result in this paper is a generalization   of  Theorem \ref{thm:fsp}  for $\spin\rightarrow +\infty$ and $\hbar\searrow 0$ such that $\spin\hbar^{\delta} =c$ where 
$0<\delta<1$ and $c >0$ are constants. Let us denote $\kappa:=\hbar\spin$ which is here a small positive parameter for $\delta<1$.
 \begin{theorem}\label{thm:wsp} Let us assume that $0<\delta<1$. 
 For the initial state $\Psi_{z_0,{\bf n}_0} =\varphi_{z_0}\otimes\psi^{(\spin)}_{{\bf n}_0}$, the solution of \eqref{eq:sch0} for the Hamiltonian \eqref{H:spin} satisfies:
\beq\label{propag_spin2}
  \cmU(t,t_0)\Psi_{z_0,{\bf n}_0}= {\rm e}^{\frac{i}{\hbar}S(t) +i\spin\alpha(t)}\varphi_{z(t)}^{\Gamma(t)}\otimes\psi^{(\spin)}_{{\bf n}(t)}
  + O(\sqrt\hbar +\kappa)),
\eeq
where the dynamics of the coherent states   satisfies the following  system of equations:
\beq\label{cl:spin0}
\begin{aligned}
\dot q &=\partial_p H_0(t,q,p)  +\kappa\partial_p\mathfrak C(t,q,p)\cdot{\bf n}(t)\\
\dot p &=- \partial_q H_0(t,q,p) -\kappa\partial_q\mathfrak C(t,q,p)\cdot{\bf n}(t)\\
\dot{\bf n}(t)&= \mathfrak C(t,q,p)\wedge{\bf n}(t).
\end{aligned}
\eeq
Moreover $S(t)$ is the action along the trajectory $z(t)$ and $\alpha(t)$ is  the action along the trajectory ${\mathbf n}(t)$
 of  the classical spin,   in the time interval $[t_0, t]$.
\end{theorem}

%\begin{remark}
%When the spin $\spin$ is constant we know from Theorem \ref{thm:fsp}  that in \eqref{cl:spin0} we have $\kappa=0$. 
% In particular the classical orbit motion that  the Gaussian  follows is spin independent. Notice that the spin motion always depends on the orbital motion.\\
%For $\kappa$ decreasing  slowly to 0: $\kappa=\hbar^{1-\delta}, 1/2 <\delta<1$, {\em the orbital coherent state follows a trajectory slightly different from the classical  trajectory for $H_0(t)$ and is depending on the spin interaction}
%(see Appendix).
%From Theorem \ref{thm:wsp} we   get that  $\spin\approx \frac{1}{\sqrt\hbar}$ is a critical regime for the quantum orbital motion.
%\end{remark}
Let be $(z^\kappa_t, {\bf n}^\kappa_t)$ the flow satisfying \eqref{cl:spin0} with $z^\kappa_0=z_0$,  ${\bf n}^\kappa_0={\bf n}_0$. 
\begin{Cor}
Let be $1/2>\varepsilon >0$ small. Then  we have \\
\begin{itemize}
\item if $\spin\leq c\hbar^{-1/2+\varepsilon}$ then  the  Theorem \ref{thm:wsp} is valid by taking $\kappa =0$ in \eqref{cl:spin0}, with the error term $O(\hbar^\varepsilon)$. So  in this case the quantum orbital follows the classical trajectory  for $H_0(t)$ at  the leading order in $\hbar$.
\item $If \spin  \approx c\hbar^{-1/2-\varepsilon}$ then the  quantum orbital motion depends on the spin motion: in general we cannot take $\kappa = 0$ in \eqref{cl:spin0}. If $\nabla_X\mathfrak C(t_0,z_0)\neq 0$  the orbital Gaussian  
$\varphi_{z(t)}^{\Gamma(t)}$  follows the new  trajectory $z^\kappa_t$ and not $z^\kappa_t$. Moreover  there exist $\tau>0$, $c>0$,  such that
\beq\label{pertraj}
\vert z^\kappa_t-z^0_t\vert \geq c\vert t-t_0\vert\hbar^{1/2-\varepsilon}, \;{\rm for}\;\; \vert t-t_0\vert\leq \tau.
\eeq
 (see Lemma \ref{pertkappa})
\end{itemize}
\end{Cor}
\begin{remark}

Notice that the coherent state $\varphi_z^\Gamma(x)$ is localized in any  disc $\{\vert X-z\vert \leq c\hbar^{1/2-\eps}\}$,
 $c>0, \eps>0$,  in the phase space $\R^{2d}$. More precisely its Wigner function 
 ${\mathcal W}_{\varphi_z^\Gamma}(X)$ satisfies for some $\mu>0$,  $C>0$, 
 $$
 \vert {\mathcal W}_{\varphi_z^\Gamma}(X) \vert \leq (\pi\hbar)^{-d}{\rm e}^{-\frac{\mu}{\hbar}\vert X-z\vert^2}.
 $$
 And   its  Husimi function ${\mathcal H}_{\varphi_z^\Gamma}(X)$  satisfies, for some $C'>0$ and $\mu'>0$, 
 $$
 0\leq {\mathcal H}_{\varphi_z^\Gamma}(X)\leq C'\hbar^{-d}{\rm e}^{-\frac{\mu'}{\hbar}\vert X-z\vert^2},
 $$
 (for  more properties of the Husimi functions see \cite{CR}, section 2.5).\\
 Hence from \eqref{pertraj} we get that ${\mathcal H}_{{\varphi_{z(t)}^{\Gamma(t)}}}(X)$ is $O(\hbar^\infty)$ in a neighborhood of size $\hbar^{1/2-\eps/2}$ of $z^0_t$ for $0<t_1<t<t+\tau$.
 \end{remark}
 Let us consider now  the case $\kappa>0$ fixed  and   $\hbar\searrow 0$.  We shall see  that Theorem \ref{thm:wsp} cannot be satisfied in general. In particular  for the Dicke model  considered by Hepp-Lieb (for $\hbar=1$) \cite{HELI} (See also \cite{RDT} for a recent review).
$$
\hat H_{\rm Dic} = \hbar\omega_c{\bf a}^\dagger{\bf a}  + \omega_3\hat S_3 + \frac{\lambda}{\sqrt{N}}({\bf a}^\dagger + {\bf a})\hat S_1,
$$ 
In  quantum optics this model   describes the interaction between light and matter  where the light (photons) is  a field operator. Here we consider a toy model with 
 an  quantum harmonic oscillator for photons  and $N$ atoms of matter being in two spin levels.\\
 $\omega_c>0$, $\omega_3\geq 0$, $N=2\spin+1$ is the dimension of the spin states space, $\lambda>0$ a coupling constant, 
${\bf a} = \frac{1}{\sqrt 2\hbar}(x+\hbar\partial_x)$, ${\bf a}^\dagger = \frac{1}{\sqrt 2\hbar}(x-\hbar\partial_x)$ are the usual creation and annihilation operators
 (satisfying $[{\bf a}, {\bf a}^\dagger]=1$) and $\hat S_j = \hbar S_j$. Here the Fock space is simply 
  $L^2(\R)$. \\

A more explicite expression  for $\hat H_{\rm Dic}$  considered here is the following 
\beq\label{Dic2}
\hat H_{\rm Dic} = \frac{\omega_c}{2}(-\hbar^2\partial_x^2 + x^2) +   \omega_3\hat S_3 +\frac{2\lambda}{\sqrt{N\hbar}}x\hat S_1.
\eeq
%Notice that the two first  terms above  are of order 0 in $\hbar$ (recall we are in the regime $N\hbar=\kappa>0$) while the third term is of order 1 (subprincipal term). 
%In particular the norm of $\hat S_1$ is of order $\sqrt{N\hbar}$.\\
With   this elementary model  we get a contradiction with  Theorem 4.7  of \cite{BG}. In the regime $\hbar \approx \spin^{-1}$,   $\spin\nearrow\infty$ 
the orbital trajectory blows up into a mixed state  for  $t>0$ ("decoherence of the orbital state"). 

%%%%%%%%%%%%%%%
%%%%%%%%%%%%%%
%The question we want to consider in this paper is the validity of  \eqref{propag_spin1},   with  a possible  modified  orbit motion depending of the spin motion,  in the regime $\hbar\spin =\kappa$ and   
%$ \hbar\rightarrow 0$.  We shall see later that this ansatz is wrong in particular for the Dicke model \eqref{Dic2} and then contradicts   Theorem 4.7  of \cite{BG}.
\begin{remark}
{\bf Why to consider large spin quantum  systems?}\\
Notice that for $N$ atoms of spin 1/2  the spin Hilbert space should be of
dimension $2^N$  corresponding to the tensor product 
$\otimes^N{\mathcal D}^{1/2}$ where   ${\mathcal D}^{1/2}$ is the representation of degree 2 of $SU(2)$  for the spin 1/2. Using the Clebsch-Gordon formula  we see that the representation  $\otimes^N{\mathcal D}^{1/2}$ contains the irreducible representation  ${\mathcal D}^{N/2}$  of $SU(2)$  of maximal degree $N+1$  corresponding to an effective spin $\spin = \frac{N}{2}$.
 So that in this settting the large spin limit is also the large number of atoms limit (thermodynamic limit).
\end{remark}
 
\begin{theorem}\label{Dmod}
Let be  $0<\kappa=\hbar\spin$,  $\Psi(0)=\varphi_{z_0}(x)\psi_{{\bf n}_0}$ and  $\Psi(t) = {\rm e}^{-\frac{it}{\hbar}\hat H_{\rm Dic}}\Psi(0)$ with  $\Psi(0)=\varphi_{z_0}\otimes\psi_{\bf n_0}$. \\
Then  there exists $c_0>0$ such that for any $0<\mu<1/2$ and $\hbar\searrow 0$, we have
\beq\label{decoh}
{\rm tr}[{\rm tr}_{{\mathcal H}_{2s+1}}\Pi_{\Psi(t)}]^2 \leq (1+c_0\kappa t^2)^{-1/2} + O(\hbar^\mu),
\eeq
 where $\Pi_{\Psi}$ is the projection on the pure state $\Psi\in L^2(\R)\otimes{\mathcal H}_{2s+1}$,   ${\rm tr}_{{\mathcal H}_{2s+1}}\Pi_{\Psi}$ is the partial trace of $\Pi_{\Psi}$ in  ${\mathcal H}_{2s+1}$ and the
 last trace in \eqref{decoh} is computed for operators in  the Hilbert space $L^2(\R,\cc)$.
\end{theorem}
\begin{remark}
The previous result  shows  that for  a spin $\spin$ of order $\hbar^{-1}$ the orbital evolution is transformed into a mixed state as the spin interaction is  switched on and $\hbar$ is small enough. \\
For  a density matrix $\hat\rho$,   ${\mathcal P}[\hat\rho]:= {\rm tr}[\hat\rho^2]$ is called  the purity of the  density matrix
 $\hat\rho$. 
Here it is applied for $\hat\rho_{\spin}(t):= {\rm tr}_{{\mathcal H}_{2s+1}}\Pi_{\Psi(t)}$. This is related with the von-Neumann entropy 
which  is defined as 
${\rm S}_{vN}[\hat\rho]=-{\rm tr}[\hat\rho\log\hat\rho]$. So we have easily  ${\rm S}_{vN}[\hat\rho]\geq 1-{\mathcal P}[\hat\rho]$. For a pure state ${\rm S}_{vN}[\hat\rho]=0$ and  $ {\mathcal P}[\hat\rho]=1$.\\
Then from \eqref{decoh} we get that the density matrix ${\rm tr}_{{\mathcal H}_{2s+1}}\Pi_{\Psi(t)}$ has a positive von Neumann entropy for $t>0$ and $\hbar>0$ small enough.\\
In the analysis of the Dicke model one consider that  the orbital  part is an open sub-system of the closed total  system (orbital+spin). For  other open systems and time  evolution of coherent states  (like  the "Schr\"odinger cat") we refer to \cite{CR}, Chapter 13 for more details.
\end{remark}
\section{Preliminaries}
\subsection{Reduction to  the interaction propagator}
It  is convenient to   annihilate the scalar part $H_0(X)$ by considering the interaction representation for the propagator $\cmU(t,t_0)$ of the Hamiltonian $\hat H(t)$.\\
Hence we have $\cmU(t,t_0) = \cmU_0(t-t_0){\mathcal V}(t,t_0)$, where 
$ \cmU_0(t-t_0) = {\rm e}^{-i\frac{t-t_0}{\hbar}\hat H_0}$ and 
the propagator ${\mathcal V}(t,t_0)$ must satisfy
$$
i\hbar\partial_t{\mathcal V}(t,t_0) = (\widehat{\mathfrak{C}}_I(t)\cdot{\bf S}){\mathcal V}(t,t_0).
$$
  For $1\leq k\leq 3$ the Weyl symbols ${\mathfrak{C}}_{I, k}(t,X)$ is computed by the Egorov  Theorem.
  In particular its principal term is given by 
  $$
   {\mathfrak{C}}_{I,k}(t,X) = {\mathfrak{C}}_k(t,\Phi_0^{t-t_0}(X)).
   $$
  So for simplicity, in what follows we shall assume that $H_0 \equiv 0$ and consider the  simpler interaction spin-orbit   Hamiltonian  
  $\hat H_{\rm int}(t) =   \hbar\widehat{\mathfrak{C}}(t)\cdot{\bf S}$.\\
  For the Dicke model  $ H_0$ is the harmonic oscillator   so we have:
  $$
  \Phi_0^t(x,\xi) =  (\cos(\omega_ct)x(0) +\sin\omega_ct\xi(0), -\sin(\omega_ct)x(0) +\cos(\omega_ct)\xi(0),
  $$
hence $H_{\rm int}(t,x,\xi) = (\cos(\omega_ct)x + \sin(\omega_ct)\xi)S_1+\hbar\omega_3S_3$.
\subsection{A realization of spin-$\spin$ representation}\label{sec:spin}

Because our aim is to perform  explicit computations for the spin side we   choose  to work with a concrete representation (see for example \cite{CR}, chap.7). \\
We assume here that ${\mathcal V}^{(\spin)} = {\mathcal H}_{2\spin+1}$, the complex linear space of homogenous polynomials of degree $2\spin$ in two variables $(z_1,z_2)\in\cc^2$. ${\mathcal H}_{2\spin+1}$ is an Hermitian space for the  scalar product defined such that 
the monomials (named Dicke states):
$$
{\mathscr{D}}^{(\spin)}_m =  \frac{z_1^{\spin+m}z_2^{\spin-m}}{\sqrt{(\spin+m)!(\spin-m)!}}, 
$$
where $-\spin \leq m\leq \spin$ and  $m$ is an integer or half an integer  according   $\spin 
$ is.\\
In ${\mathcal H}_{2\spin}$ the spin operators are realized  as
\bea
S_3 &=&\frac{1}{2}\left(z_1\partial_{z_1}- z_2\partial_{z_2}\right),\; S_+ = z_1\partial_{z_2},\; S_- = z_2\partial_{z_1}\nonumber\\ S_1&=&\frac{S_++S_-}{2},\; S_2=\frac{S_+-S_-}{2i},
\eea
with the commutation relations of the Lie algebra $su(2)$ of the group $SU(2)$:
$$
[S_3, S_\pm] = \pm S_\pm.
$$
Recall that $g\in SU(2)$ if $g=\begin{pmatrix}\alpha & -\bar\beta\\ \beta & \alpha\end{pmatrix}$,\;$ \alpha, \beta \in\cc$, $\vert\alpha\vert^2+\vert\beta\vert^2=1$.\\
Spin coherent states are defined by the action of $SU(2)$ on the Dicke state 
 of minimal weight: ${\mathscr D}^{(\spin)}_{-\spin}$. 
 As known, to get a good parametrization of coherent states we choose a family of transformations in $SU(2)$ indices by the 
sphere $\Sp^2$. This is obtained  by computing the isotropy subgroup ${\rm ISO}$ of the $SU(2)$ action:
$g\in {\rm ISO}$ if and only if $g=\begin{pmatrix}\alpha & 0\\ 0&\bar\alpha\end{pmatrix}$, $\alpha={\rm e}^{i\psi}, \psi\in\R$.
Then we get that the space of orbits $SU(2)\backslash{\rm ISO}$ can be  identified with the sphere $\Sp^2$ 
(for details see for example \cite{CR}). 

Let us denote
 %Choosing the following parametrization of $\Sp^2$
 
% \alpha=\cos\frac{\theta}{2},\quad \beta=-\sin\frac{\theta}{2}e^{-i\varphi},\quad 0\le \theta< 2\pi,\quad 0\le \varphi< 2\pi$$
% we see that the space $X_0$ is just a representation of the two-dimensional sphere $\Sp^2$  minus the north pole, namely the set of unit three-dimensional vectors
 $$
 \mathbf n=(\sin\theta\cos\varphi,\sin\theta\sin\varphi,\cos\theta),\; 0 \leq\theta <\pi,\; 0\leq\varphi<2\pi;
 $$
To any $\mathbf n\in\Sp^2$ we associate the transformation in $SU(2)$,
 \beq\label{expn}
 g=g_{\mathbf n}=\exp\left(i\frac{\theta}{2}(\sin\varphi\sigma_1 -\cos\varphi\sigma_2)\right)
 \eeq
 where  $\sigma_1,\ \sigma_2$ are the Pauli matrices which satisfy the commutation relations
\beq
[\sigma_k,\sigma_l]=2i\epsilon_{k,\ell,m}\sigma_m
\eeq
For $\spin=\frac{1}{2}$ we have $S_k =\frac{\sigma_k}{2}$. So if $g\mapsto {\mathcal D}^{(\spin)}g $ is the representation of $SU(2)$ in ${\mathcal H}_{2\spin+1}$ we have 
\beq\label{act}
 {\mathcal D}^{(\spin)}g_{\mathbf n} = \exp\left(\frac{\theta}{2}({\rm e}^{i\varphi}S_- -{\rm e}^{-i\varphi}S_+)\right)
 \eeq
 %Thus $g_{\mathbf n}$ describes a rotation by the angle $\theta$ around the vector 
 % $\mathbf m= (\sin\varphi, -\cos\varphi,0)$ belonging to the equatorial plane of the sphere and perpendicular to $\mathbf n$.\\
 \begin{Def}
 The coherent states of $SU(2)$ are  defined  in the representation space ${\mathcal H}_{2\spin+1}$ as follows:
 \beq
 \vert \mathbf n\rangle ={\mathcal D}^{\spin}(g_{\mathbf n}){\mathscr{D}}^{(\spin)}_{-\spin} :=\psi_{\mathbf n}^{(\spin)}.
 \eeq
\end{Def}

It is some time convenient to consider a complex parametrization of the sphere $\Sp^2$ using the stereographic projection ${\mathbf n}\mapsto \eta$,  on the complex plane and an identification of the coherent states $\vert \mathbf n\rangle$ with the state  $\vert \eta\rangle:=\psi_\eta$ defined as follows:
 $$
 \vert \eta \rangle= (1+\vert \eta \vert^2)^{-j}\exp(\eta S_+)\vert \spin,-\spin\rangle,
 $$
%where $S_+=S_1+iS_2$, $S_- = S_1-iS_2$.\\
  More precisely we  denote $\vert\eta\ra = \vert{\mathbf n}\ra$ with  the following correspondence:
 $$
 \mathbf n=(\sin\theta\cos \varphi, \sin\theta \sin\varphi, \cos \theta),\quad \eta= -\tan\frac{\theta}{2}e^{-i\varphi}
 $$
 The geometrical interpretation is that $-\bar\eta$ is the stereographic projection of $\bf n$.\\
 Recall the following expression of $g_{\mathbf n}$
 \beq
 g_{\mathbf n} = \begin{pmatrix} \cos\frac{\theta}{2} & -\sin\frac{\theta}{2}{\rm e}^{-i\varphi}\\ 
 \sin\frac{\theta}{2}{\rm e}^{i\varphi} & \cos\frac{\theta}{2}\end{pmatrix}
 \eeq
  For further investigations  we shall need two  results:\\
  1) compute the derivative of $t\mapsto T(g_{{\mathbf n}_t})$   for a $C^1$ path on $\Sp^2$.\\
  2) compute the adjoint action of $T(g_{\mathbf n})$ on $S_k$, $1\leq k\leq 3$.\\
  From \eqref{act} we get
  \bea\label{der1}
  \partial_\varphi T(g_{\mathbf n} )& =& \frac{i}{2}T(g_{\mathbf n})\left(\sin\theta({\rm e}^{-i\varphi}S_+ + {\rm e}^{i\varphi}S_-) + (1-\cos\theta)S_3 \right) \\
    i\partial_\theta T(g_{\mathbf n} )&=& \frac{i}{2}T(g_{\mathbf n})({\rm e}^{i\varphi}S_- - {\rm e}^{-i\varphi}S_+)
  \eea  
  \begin{lem}\label{derivsp}
  Let be a $C^1$ path on $\Sp^2$: $t\mapsto (\theta_t, \varphi_t)$. Then we have
  \bea
  &\partial_tT(g_{\mathbf n}) = \frac{i}{2}T(g_{\mathbf n})\left(\sin\theta({\rm e}^{-i\varphi}S_+ + {\rm e}^{i\varphi}S_-) + (1-\cos\theta)S_3\right)\dot\varphi_t \nonumber\\
  &+ \frac{1}{2}T(g_{\mathbf n})\left({\rm e}^{i\varphi}S_- - {\rm e}^{-i\varphi}S_+\right)\dot\theta_t.
  \eea
  \end{lem}
  %\textcolor{red}{check the  previous lemma!}
  To compute $\dot\theta$, $\dot\varphi$ we use the Riemann model for $\Sp^2$. Let be $-\bar\eta\in\cc$ the stereographic projection of ${\mathbf n}$ on the equatorial plane, from the south pole  ${\mathbf n}_{so}$. The coordinates of ${\mathbf n} $ are given by
  $$
  n_1=-\frac{\eta +\bar\eta}{1+\vert\eta\vert^2},\;\;n_2 =\frac{\eta -\bar\eta}{i(1+\vert\eta\vert^2)},\;\; n_3= \frac{1-\vert\eta\vert^2}{1+\vert\eta\vert^2}.
  $$
  
  So we have
  \beq\label{ccoordin}
  \eta =-\tan\left(\frac{\theta}{2}\right)\,{\rm e}^{-i\varphi},\;\cos\theta = \frac{1-\vert\eta\vert^2}{1+\vert\eta\vert^2},\; \sin\theta\,{\rm e}^{i\varphi}=-\frac{2\,\bar\eta}{1+\vert\eta\vert^2}. 
  \eeq
  Then we get easily
  \bea\label{deriv2}
  \dot\theta = \frac{\eta\dot{\bar\eta} + \bar\eta\dot{\eta} }{\vert\eta\vert(1+\vert\eta\vert^2)},\;\; 
  \dot\varphi =\frac{i}{2}\left(\frac{\dot\eta}{\eta} - \frac{\dot{\bar\eta}}{\bar\eta}\right).
   \eea
   %\textcolor{red}{compute $i\partial_tT(g_{{\mathbf n}_t})$. see ipad. p.39-40}
   
   \begin{lem}\label{lem:derivsp2}
   $$
-i\partial_tT(g_{{\mathbf n}_t})= iT(g_{{\mathbf n}_t})(A\dot\theta +B\dot\varphi)
$$
with
$$(A\dot\theta +B\dot\varphi) = -\frac{1}{1+\vert\eta\vert^2}(\dot\eta S_+ -\dot{\bar\eta}S_-) 
+ \frac{\vert\eta\vert^2}{1+\vert\eta\vert^2}\left(\frac{\dot\eta}{\eta} - \frac{\dot{\bar\eta}}{\bar\eta}\right)S_3
$$
   \end{lem}
   {\em Proof}. Standard computations. $\square$
  %%%%%%%%%%%%%%%%%%%%%%%%%%%%%%
  %%%%%%%%%%%%%%%%%%%%%%%%%%%%%%
  
  Let us denote $S_k(\mathbf n) = {\mathcal D}^{\spin}(g_{\mathbf n})^*S_k {\mathcal D}^{\spin}(g_{\mathbf n})$, $k=+, - , 3$ or $k=1,2,3$. 
  \\In the following Lemma similar results are stated in \cite{ FKL}, section 2.6,   and proved by a different method.
  \begin{lem}\label{lem:adS}
 If $(\theta, \varphi)$ are the coordinates of $\mathbf n$  on $\Sp^2$, we have,
 \bea\label{Spm}
 S_3(\theta,\varphi) &=& \cos\theta.S_3 -\sin\theta\left(\frac{ {\rm e}^{i\varphi}S_- + {\rm e}^{-i\varphi}S_+}{2}\right) \\ 
 S_+(\theta, \varphi) &=&{\rm e}^{i\varphi}\sin\theta. S_3  + \frac{\cos\theta +1}{2}S_+ + \frac{\cos\theta -1}{2}{\rm e}^{2i\varphi}S_- \\
  S_-(\theta, \varphi) &=&{\rm e}^{-i\varphi}\sin\theta. S_3  + \frac{\cos\theta +1}{2}S_- + \frac{\cos\theta -1}{2}{\rm e}^{-2i\varphi}S_+ 
  \eea
  \bea\label{S12}
   &S_1(\theta, \varphi) = \left(\frac{\cos\theta +1}{2} + \frac{\cos\theta -1}{2}\cos 2\varphi \right) S_1 + \frac{\cos\theta -1}{2}\sin 2\varphi \,S_2  \nonumber\\
  & + \cos\varphi\sin\theta \;S_3\nonumber\\
  \eea
  \bea
   &S_2(\theta, \varphi) = \frac{\cos\theta -1}{2}\sin 2\varphi S_1 + 
   \left(\frac{\cos\theta +1}{2} -\frac{\cos\theta -1}{2}\cos 2\varphi\right)S_2\nonumber\\
 &+\sin\theta\sin\varphi S_3.
  \eea 
 \end{lem}
  {\em Proof}.  We write  ${\mathcal D}^{\spin}(g_{\mathbf n}) =   {\rm e}^{\frac{\theta}{2}L}$ where $L= {\rm e}^{i\varphi}S_--{\rm e}^{-i\varphi}S_+$.
  Let be $S$ one of the spin operator, we have 
  $$
  \partial_\theta S(\theta,\varphi) = \frac{1}{2} {\rm e}^{-\frac{\theta}{2}L}[S,L] {\rm e}^{\frac{\theta}{2}L}
  $$
    and  the commutation relations
  $$
  [L,S_3] =  {\rm e}^{i\varphi}S_ + +{\rm e}^{-i\varphi}S_+, \;[L,S_-] = -2{\rm e}^{-i\varphi}S_3,\;  [L,S_+] = -2{\rm e}^{i\varphi}S_3.
  $$
  So we find
  \bea
 \partial_\theta S_3(\theta,\varphi) &=& -\frac{1}{2}( {\rm e}^{i\varphi}S_-(\theta,\varphi) +  {\rm e}^{-i\varphi}S_+(\theta,\varphi))\\
 \partial_\theta S_+(\theta, \varphi) &=& {\rm e}^{i\varphi}S_3(\theta,\varphi) \\
  \partial_\theta S_-(\theta, \varphi) &=& {\rm e}^{-i\varphi}S_3(\theta,\varphi).
  \eea
hence
  \beq\label{S3}
  \partial_\theta^2S_3(\theta,\varphi)=-S_3(\theta,\varphi).
  \eeq 
  Solving the differential equation \eqref{S3} we get \eqref{Spm}.
  $\square$
  
  Now we come  to the spin-coherent states. Le be ${\mathbf n}_0$ the north pole on $\Sp^2$  corresponding to the Dicke state 
  ${\mathscr{D}}^{(\spin)}_{-\spin} :=\psi_{{\mathbf n}_0}$.  The following Lemma  is basic for our  next computations.% and contradicts  Lemma 4.5 of \cite{BG}.
  \begin{lem}\label{lem:sqs}
   For any $\mathbf n\in\Sp^2$ we have 
   $$
   {\mathbf S}\psi_{\mathbf n} = -\spin{\mathbf n}\psi_{\mathbf n}  + \sqrt{\frac{\spin}{2}}\,{\mathbf v}({\mathbf n})\psi_{1,{\mathbf n}}
   $$
   where $\psi_{1,{\mathbf n}}={\mathcal D}^{(\spin)}(g_{\mathbf n}){\mathscr{D}}^{(\spin)}_{1-\spin}$ and ${\mathbf v}({\mathbf n})=(v_1, v_2, v_3)\in\cc^3$ is defined as 
   \bea
   v_1(\mathbf n) &=& \frac{\cos\theta +1}{2} + \frac{\cos\theta -1}{2}{\rm e}^{-2i\varphi}\\
    v_2(\mathbf n) &=& \frac{\cos\theta +1}{2i}  - \frac{\cos\theta -1}{2i}{\rm e}^{-2i\varphi}\\
    v_3(\mathbf n) &=& -\sin\theta{\rm e}^{-i\varphi}
       \eea
In particular we have ${\mathbf n}\cdot{\mathbf v}(\mathbf n)=0$.
  \end{lem}
{\em Proof}.  
The formulas are proved  from elementary computations using Lemma \ref{lem:adS} and that $S_3\psi_{{\mathbf n}_0}=-\spin\psi_{{\mathbf n}_0}$, 
$S_+\psi_{{\mathbf n}_0}=\sqrt{2\spin}\,\psi_{1,{\mathbf n}_0}$, $S_-\psi_{{\mathbf n}_0}=0$. In particular we recover here a well 
known property of the spin coherent states: \\
$${\mathbf n}\cdot{\mathbf S}\psi_{\mathbf n} = -\spin\psi_{\mathbf n}.$$\\
$\square$
                                                
\begin{remark}
In Lemma \ref{lem:sqs}  the formulas have only two terms, a leading term of order $\spin$ and a second term of order $\sqrt{\spin}$. In Lemma 4.5 of \cite{BG} the authors claims 
that the expansion is  an asymptotic  power serie in $\spin^{-1}$.  This contradicts   our  computations.
\end{remark}                  
It is also convenient to write down the evolution of the  spin matrices on the Riemann sphere, denoting $S_k(\eta) =S_k(\mathbf n)$
where $\eta\in\cc$ is the complex coordinate of ${\mathbf n}\in\Sp^2$.
\begin{lem}\label{lem:scc}
\bea
S_3(\eta) &=& \frac{1-\vert\eta\vert^2}{1+\vert\eta\vert^2}S_3 + \frac{\eta S_+ +\bar\eta S_-}{1+\vert\eta\vert^2},\nonumber\\
S_+(\eta) &=&-2\frac{\bar\eta}{1+\vert\eta\vert^2}S_3 + \frac{1}{1+\vert\eta\vert^2}S_+ - \frac{\eta^2}{1+\vert\eta\vert^2}S_-\nonumber\\
S_-(\eta) &=&-2\frac{\eta}{1+\vert\eta\vert^2}S_3 + \frac{1}{1+\vert\eta\vert^2}S_- - \frac{{\bar\eta}^2}{1+\vert\eta\vert^2}S_+   \nonumber.
\eea
\end{lem}
%%%%%%%%%%%%%%%%%%%%%%%%%%%%%%%%%%
%%%%%%%%%%%%%%%%%%%%%%%%%%%%%%%%%%%%%%
\subsection{The classical spin space}\label{clspin}
Let us recall that  the sphere $\Sp^2$ has a natural symplectic form:\\$\sigma_{\bf n}(u, v) = (u\wedge v)\cdot{\bf n}=\det(u,v,\bf n)$, where ${\bf n}\in\Sp^2$, $u,v\in T_{\bf n}(\Sp^2)$.\\
Let be $H:\R^3\rightarrow \R$ a smooth function. Its resriction to $\Sp^2$ defined an Hamiltonian vector field  $X_H$  on $\Sp^2$ satisfying 
$dH(Y) = \sigma_{\bf n}(Y,X_H),\; \forall Y\in T_{\bf n}(\Sp^2)$. So we get the Hamilton equation (named in this context the Landau equation):
$$
\dot{\bf n} = \nabla H\wedge\bf n.
$$
In complex coordinates \eqref{ccoordin} the covariant symbol 
$H_c(t,\eta, \bar\eta)=\langle\psi_\eta, \hat H(t)\psi_\eta\rangle$  of the Hamiltonian  $\hat H(t) = \hbar{\mathfrak{C}}(t)\cdot{\bf S}$, becomes
$$
H_c(t,\eta, \bar\eta) = (1+\vert\eta\vert^2)^{-1}(\mathfrak C_3(t)(1-\vert\eta\vert^2) -({\mathfrak C}_-(t)-\bar\eta +
{\mathfrak C}_-(t)+\eta))
$$
where $\mathfrak C_\pm = \mathfrak C_1 \pm i\mathfrak C_2$. Recall that the symplectic form on the Riemann sphere
 $\hat\cc$  is
$\sigma_c = 2i(1+\vert\eta\vert^2)^{-2}d\eta\wedge d\bar\eta$. So the Hamilton  equation in $\hat\cc$  becomes
\beq\label{clandau}
\dot\eta = \frac{(1+\vert\eta\vert^2)^2}{2i}\partial_{\bar\eta}H_c(t, \eta,\bar\eta).
\eeq
%\textcolor{red} {Check}\\
This is the Landau-Lifschitz equation $ \eqref{eq:LL}$ in complex coordinates.\\
Following \cite{Ono} the symplectic two form satisfies $d{\mathcal \theta}_c = \sigma_c$ where the one-form  is 
$\mathbf{\theta}_c = i^{-1}\left(\partial_{\bar\eta}Kd\bar\eta-\partial_{\eta}Kd\eta\right)$ and $K(\eta,\bar\eta)=2\log(1+\eta\bar\eta)$ is the K\"ahler potential for
 $\hat\cc$.\\
In particular  the action $\Gamma$  for $H_c$ is the one form in $\hat\cc_\eta\times\R_t$ 
satisfying  \\$d\Gamma_c = \mathbf{\theta}_c -  H_cdt$.
So, along an Hamiltonian path in time interval $[0,T]$, the action  is given by (see  Appendix): 
\beq\label{action}
\gamma(T) = \int_0^T\left( \frac{\Im(\eta_t\dot{\bar\eta}_t)}{2(1+\vert\eta_t\vert^2)}-H_c(\eta_t, \bar\eta_t) \right)dt
\eeq

In the same Appendix it is proved that  $\alpha(t) = \gamma(t)$  where $\alpha(t)$ is the phase given in  Theorem \ref{thm:wsp}.

 \subsection{The Schr\"odinger  coherent states}

  We shall use some  well known formulas concerning the Heisenberg translations  operators $\hat T(z)$ and coherent states.
  \begin{lem}
  Le be $t\mapsto z_t =(q_t, p_t)$ a $C^1$ path in the phase space $\R^{2d}$. Then we have
  \beq\label{derivheis}
  i\hbar\partial_t \hat T(z_t) =\hat T(z_t)\left(\frac{1}{2}\sigma(z_t, \dot{z}_t)+\dot q_t\cdot\hbar\nabla_x -\dot p_t\cdot x\right).
  \eeq
       \end{lem}        
       \begin{lem}\label{pseudo_coh}\cite[chap.2]{CR}
    Assume  that  $A$ a  sub-polynomial symbol.   Then  for every $N\geq 1$, we have
\beq
\hat A\varphi_z =
\sum_{\vert\gamma\vert\leq N} \hbar^{\frac{|\gamma|}{2}}\frac{ \partial^{\gamma}
A(z)}{\gamma!}
\Psi_{\gamma, z} + {\mathcal O}(\hbar^{(N+1)/2}),
\eeq
the estimate of the remainder is uniform  in  $L^2(\R^n)$  for $z$ in every bounded set of the phase space and 
\beq\label{truc}
\Psi_{\gamma ,z} = 
\hat T(z ) \Lambda_{\hbar} {\rm Op}^w_1(z^\gamma)g. 
\eeq
 where $g(x) = \pi^{-n/4}{\rm e}^{-\vert x\vert^2/2}$,   ${\rm Op}^w_1(z^\gamma)$ is the 1-Weyl quantization of the monomial~:\\
$(x, \xi)^{\gamma} = x^{\gamma^\prime}\xi^{\gamma^{\prime\prime}}$, 
\hbox{$\gamma = (\gamma^\prime, \gamma^{\prime\prime})\in\n^{2d}$}.
In particular ${\rm Op}^w_1(z^\gamma)g = P_\gamma g$ where $P_\gamma$ is a polynomial 
 of the same parity as $\vert\gamma\vert$.
 \end{lem}                                                                              
   %%%%%%%%%%%%%%%%%%%%%%%%%%%%%%%%%%%%%%%%%%%%%%%%%%%%%%%%%%%%%%%%
  %%%%%%%%%%%%%%%%%%%%%%%%%%%%%%%%%%%%%%%%%%%%%%%%%%%%
  \section{The  Schr\"odinger equation    for the spin-orbit interaction}
  Recall our reduced Hamiltonian  $\hat H(t) = \hbar\widehat{\mathfrak{C}}(t)\cdot{\bf S}$ and the Schr\"odinger equation
\beq\label{eq:sch} ( i\hbar\partial_t-\hat H(t))\Psi(t) =0,\;
 \Psi(0) = \varphi_{z_0}\otimes\psi_{{\mathbf n}_0}.
 \eeq
 Let us consider the following ansatz 
 $$
 \Psi_{\rm app}(t) = {\rm e}^{\frac{i}{\hbar}\gamma(t)}\varphi_{z(t)}^{\Gamma(t)}\otimes\psi^{(\spin)}_{{\bf n}(t)}
 $$
  such  that  for $\hbar^\delta\spin=\kappa>0$  and  some $\mu>0$ we have: 
  \beq\label{ans}
  ( i\hbar\partial_t-\hat H(t))\Psi_{\rm app}(t) =  O(\hbar^{1+\mu})\;\;{\rm for}\;\;\hbar\rightarrow 0
  \eeq
  hence  from Duhamel formula we should have
  $$
\Vert \Psi(t)-\Psi_{\rm app}(t)\Vert = O(\hbar^{\mu}).
$$
If $\mathfrak{C}(t)$ depends only on time $t$ then the ansatz $\eqref{ans}$ gives  an exact  solution for some $\gamma(t)$, ${\mathbf n}_t$
 \cite{BG}. 
%%%%%%%%%%%%%%%%%%%%%%%%%%%
%%%%%%%%%%%%%%%%%%%%%%%%%%%%%
\begin{Pro}\label{tind} Let us assume that $\hat H(t) = \hbar{\mathfrak{C}}(t)\cdot{\bf S}$, where  ${\mathfrak{C}}(t)$  depends only on time. 
Then the ansatz \eqref{ans} is exact for any  $\spin$.
More explicitly we have
\beq\label{no}
 \Psi_{\rm app}(t) = {\rm e}^{i\spin\alpha(t)}\varphi_{z_0}\otimes\psi^{(\spin)}_{{\bf n}(t)}
 \eeq
The classical spin  motion $\bf{n(t)}$  follows the Landau-Lifshitz  equation
\beq\label{L}
\dot{\mathbf n}_t = {\mathfrak C}(t)\wedge{\mathbf n}_t.
\eeq
and the phase $\alpha(t)$ is the classical spin action computed in \eqref{action}.
\end{Pro}
{\em Proof.} This is well known (see a proof in Appendix \ref{pftind}).
%$\square$
\section{The regime $\spin = c\hbar^{-\delta}$, $0<\delta < 1$}\label{anstrue}
In this section  we give a proof of Theorem \ref{thm:wsp}. We have two small parameters $\hbar$ and $\kappa:=\spin\hbar = c\hbar^{1-\delta}$.\\
In this regime we shall prove that the formula \eqref{propag_spin1} is still valid with the error estimate 
$O(\sqrt\hbar + \kappa)$.\\
Let us consider the ansatz  
\beq\label{ansatz}
\Psi^\sharp(t) ={\rm e}^{\frac{i}{\hbar}S(t,t_0) +i\spin\alpha(t)}\varphi_{z(t)}^{\Gamma(t)}\otimes\psi^{(\spin)}_{{\bf n}(t)}.
\eeq
As for the scalar Schr\"odinger equation (see \cite{CR}, Chap. 4) we shall compute   a  path $t\mapsto (z(t), {\mathbf n}(t))$, a covariance matrix $\Gamma(t)$ and   phases $S(t, t_0), \alpha(t,t_0)$ such that \eqref{propag_spin1} is satisfied.\\
It is enough to prove the result for the interaction  propagator ${\mathcal V}(t,t_0)$ such that the orbital Hamiltonian 
 $H_0(t)$ is absent.
The full result is obtained by applying the scalar propagator of the orbital motion.\\
 It is enough to prove, for the norm in $L^2(\R^d, {\mathcal H}_{2\spin+1})$, we have
$$
i\hbar\partial_t\Psi^\sharp(t) = \hat H(t)\Psi^\sharp(t) +  O(\hbar(\sqrt\hbar +\kappa)).
$$
By standard computations we get  asymptotic expansions in $\hbar$ for each term  $i\hbar\partial_t\Psi^\sharp(t)$ and  $\hat H(t)\Psi^\sharp(t)$ 
mod an error $O(\hbar(\sqrt\hbar +\kappa))$. \\ 
Recall the notations 
\bea
\varphi^\Gamma_z &= &\hat T(z)\Lambda_\hbar g^\Gamma_0, \;  g_0(x) = (\pi)^{-d/4}c_\Gamma{\rm e}^{i<x,\Gamma x>}, \Lambda_\hbar f(x)=\hbar^{-d/4}f(x/\sqrt\hbar) \;\; {\rm and}\nonumber\\
\psi_{\bf n}&=&T(g_{\mathbf n})\psi_{{\bf n}_0}.
\eea
Using \eqref{derivheis}, \eqref{derivsp},  Lemma \ref{lem:derivsp2}, Lemma \ref{pseudo_coh},
 and Taylor formula around $z_t$ at the order 3,  the Schr\"odinger equation for the ansatz  \eqref{ans}   is transformed  as follows.
\begin{align}\label{anslin}
&\left(\frac{1}{2}\sigma(z_t,\dot z_t)-\partial_tS(t) -\spin\hbar\dot\alpha(t) +\sqrt\hbar (\dot q_t\cdot\nabla_x -\dot p_t\cdot x)\right)g_0\otimes\psi_{{\mathbf n}_0}  + 
\\ & g_0\otimes\hbar(A\dot\theta + B\dot \varphi)\psi_{{\mathbf n}_0} \nonumber\\
= & \sum_{1\leq k\leq 3}\left(\mathfrak{C}_k(t,z_t) +\sqrt\hbar\nabla_X\mathfrak{C}_k(t,z_t){\rm Op}^w_1(X)\right)g_0\otimes\hbar S_k({\mathbf n}_t)\psi_{{\mathbf n}_0}+ O(\hbar(\sqrt\hbar +\kappa))\nonumber.
\end{align}
But we have 
$$
(\dot q_t\cdot\nabla_x -\dot p_t\cdot x)g_0= (i\dot q -\dot p)\cdot xg_0,\;\; {\rm  and }\;\; {\rm Op}^w_1(a\cdot X)g_0 = (\alpha+i\beta)\cdot xg_0
$$
where $a\cdot X=\alpha\cdot x +\beta\cdot\xi$.
 From Lemma \ref{lem:derivsp2} we have 
 \beq\label{eq:3.3}
A\dot\theta +B\dot\varphi) = -\sqrt{2\spin}\frac{\dot\eta_t}{1+\vert\eta_t\vert^2}\psi_{1,{\mathbf n}_0}  + 
 \spin\frac{\vert\eta_t\vert^2}{1+\vert\eta_t\vert^2}\left(\frac{\dot\eta_t}{\eta_t} - \frac{\dot{\bar{\eta_t}}}{\bar\eta_t}\right)\psi_{{\mathbf n}_0}.
\eeq
From Lemma \ref{lem:scc} we have also
\bea\label{schs}
S_3(\eta)\psi_{{\mathbf n}_0} &=& -\spin\frac{1-\vert\eta\vert^2}{1+\vert\eta\vert^2}\psi_{{\mathbf n}_0} + \sqrt{2\spin}\frac{\eta}{1+\vert\eta\vert^2}
\psi_{1,{\mathbf n}_0},\nonumber\\
S_+(\eta)\psi_{{\mathbf n}_0}&=&2\spin\frac{\bar\eta}{1+\vert\eta\vert^2}\psi_{{\mathbf n}_0}+\sqrt{2\spin} \frac{1}{1+\vert\eta\vert^2}
\psi_{1,{\mathbf n}_0}\nonumber\\
S_-(\eta)\psi_{{\mathbf n}_0} &=&2\spin\frac{\eta}{1+\vert\eta\vert^2}\psi_{{\mathbf n}_0}  - \sqrt{2\spin}\frac{{\bar\eta}^2}{1+\vert\eta\vert^2}\psi_{1,{\mathbf n}_0}\nonumber
\eea

%We get for terms corresponding to $\hbar^0, \hbar^{1/2}, \hbar^{1-\delta/2}, \hbar^{1-\delta}$ 
Now we get the equations to compute the ansatz  by identifying the coefficients of  $(\hbar, \kappa)$  in \eqref{anslin}.
Easy computations give the following results.
\begin{itemize}
\item Projection on $g_0\otimes\psi_{{\mathbf n}_0}$:  the coefficient of $\hbar^0$ gives   the classical action $S(t,t_0)$ and the coefficient of $\kappa$ gives 
   the spin action $\alpha(t)$. In particular. we get
   \beq\label{alpha}
   \dot\alpha(t) = \frac{\Im(\eta_t\dot{\bar\eta}_t)}{2(1+\vert\eta_t\vert^2)} -H_c(\eta_t, \bar\eta_t),
   \eeq
   where $H_c(\eta, \bar\eta)=\la\psi_{\eta},{\mathfrak C}(t)\cdot{\mathbf S}\psi_{\eta}\ra$. 
\item Projection on $xg_0\otimes\psi_{{\mathbf n}_0}$ determines the Hamilton  equation for the orbit $z_t$.
\item projection on $xg_0\otimes\psi_{1,{\mathbf n}_0}$ determines the  Landau-Lifshitz  equation for ${\mathbf n}(t)$.
\end{itemize}

Notice that for the interaction dynamics the covariant matrix $\Gamma$ is constant, it depends only on the orbital dynamics for $H_0(t) $ not in the interaction with the spin,   contrary to the orbital motion when  $\kappa \gg \sqrt\hbar$.\\
Using the same method as in the scalar case considered in \cite{CR}, Chapter 4 we can complete the proof of Theorem \ref{thm:wsp}.
%%%%%%%%%%%%%%%%%%%%%%%%
%%%%%%%%%%%%%%%%%%%%%%%%%%
\section{The regime $\hbar\spin=\kappa$=constant} \label{crit}
The ansatz \ref{ans} is very natural when considering the classical analogue of \eqref{eq:sch}. We assume  $\hbar\spin=\kappa>0$  and we keep $\hbar$ as our semi-classical parameter. For simplicity assume that $\kappa=\frac{1}{2}$.
\\The symbol of $\hat H(t)=\widehat{\mathfrak{C}}(t)\cdot\hat{\bf S}$,
  is $H(t,q,p;{\bf n}) = -\kappa\mathfrak C(t,q,p)\cdot{\bf n}$, defined on $T^*(\R)\times\Sp^2$, (recall that the covariant symbol of ${\bf S}$ is $-s{\bf n}$, see \cite{CR}, prop. 90) we get the classical system of equations
\bea
\dot q &=& -1/2\partial_p\mathfrak C(t,q,p)\cdot{\bf n},\nonumber\\
\dot p &=& 1/2\partial_q\mathfrak C(t,q,p)\cdot{\bf n},\nonumber\\
\dot{\bf n}&=& \mathfrak C(t,q,p)\wedge{\bf n}.
\eea
When $\mathfrak C(t)$ depends only on time there is no orbit interaction with the spin and using  section \ref{sec:spin} we  can compute the phase $\alpha(t)$ such that \eqref{ans}
is the exact solution of  \eqref{eq:sch}. 
$\square$
\\
{\em But  in the following computations we shall see that \eqref{ans} is not possible if \\ $\nabla_X{\mathfrak C}(t,X)\neq 0$}.\\
For simplicity we assume here  that for $1\leq k\leq 3$,  $\mathfrak{C}_k(t)$ is a linear form on $\R^{2d}$, 
$\mathfrak{C}_k(t,X)=a_k(t)\cdot X= \alpha_k(t)\cdot x + \beta_k(t)\cdot \xi$,\; $X=(x,\xi)$.\\
%%%%%%%%%%%%%%%%%%%%%%%%%%%
%%%%%%%%%%%%%%%%%%%%%%%%%%%
%%%%%%%%%%%%%%%%%%%%%%%%%%
%%%%%%%%%%%%%%%%%%%%%%%%%%%%
Let us revisit the computations of Section \ref{anstrue} in the particular case.
 Denote ${\mathfrak C}_\pm = {\mathfrak C}_1 \pm i{\mathfrak C}_2$.\\
 Let us denote $(0)_L$ and $(0)_R$ the coefficient of $\hbar^0$ on left and right side of \eqref{anslin}. In the same way we introduce $(1/2)_{L,R}$ and 
 $(1)_{L,R}$. for the coefficients of $\hbar^{1/2}$ and $\hbar^1$.  Just compute to get
\bea\label{0}
 (0)_L=\frac{1}{2}\sigma(z_t, \dot z_t)-\dot\gamma(t) -\frac{1}{2}\frac{\vert\eta_t\vert^2}{1+\vert\eta_t\vert^2}\left(\frac{\dot\eta_t}{\eta} - \frac{\dot{\bar{\eta_t}}}{\bar\eta_t}\right)\\
 (0)_R = -\frac{1-\vert\eta_t\vert^2}{1+\vert\eta_t\vert^2}{\mathfrak C}_3(t,z_t) -\frac{{\bar\eta}^2}{1+\vert\eta\vert^2}{\mathfrak C}_-(t,z_t) - 
 \frac{{\bar\eta}^2}{1+\vert\eta\vert^2}{\mathfrak C}_+(t,z_t) \nonumber
 \eea
 The equation $(0)_L=(0)_R$ determine $\gamma(t)$ when $z_t$, $\eta_t$ are known.\\
 In  the linear case consider here $\nabla_X\mathfrak{C}(t)$ is independent on $X$. So we have:
 \bea\label{1/2}
(1/2)_L &=& (-i\dot q -\dot p)xg_0\otimes\psi_{{\mathbf n}_0}  - \frac{\dot\eta_t}{1+\vert\eta_t\vert^2}g_0\otimes\psi_{1,{\mathbf n}_0} \nonumber\\
(1/2)_R &=& -\left(\frac{1-\vert\eta_t\vert^2}{1+\vert\eta_t\vert^2}(\alpha_3(t) +i\beta_3(t)) 
+\frac{\eta_t}{1+\vert\eta_t\vert^2}(\alpha_+(t) +i\beta_+(t))\right. \nonumber\\
 &+& \left.\frac{\bar\eta_t}{1+\vert\eta_t\vert^2}(\alpha_-(t) +i\beta_-(t))\right)\cdot xg_0\otimes\psi_0 \nonumber\\
&& +  \left({\mathfrak C}_3(t)\frac{\eta_t}{1+\vert\eta_t\vert^2} +    {\mathfrak C}_-(t)\frac{1}{1+\vert\eta_t\vert^2} \right.\nonumber\\
&&- \left.{\mathfrak C}_+(t)\frac{{\bar\eta}^2_t}{1+\vert\eta_t\vert^2} \right)
g_0\otimes \psi_{1,{\mathbf n}_0} 
\eea
We denote ${\mathfrak C}(t) := {\mathfrak C}(t,z_t)$.
From  the equation $(1/2)_L = (1/2)_R$, using that 
the states $xg_0\otimes\psi_0$ and   $g_0\otimes \psi_{1,{\mathbf n}_0}$  are orthogonal  in ${\mathcal H}_{2\spin}$,  we obtain a system of coupled equations which determines a  trajectory $(z_t, \eta_t)$ in $\R^{2d}\times\Sp^2$. In particular   we get again for ${\mathbf n}_t$ the Landau equation \eqref{L} for ${\mathfrak C}(t, z_t)$ but here the time dependent equation for $z_t$ depends on ${\mathbf n}_t$.
%%%%%%%%%%%%%%%%%%%%%%%%%
%%%%%%%%%%%%%%%%%%%%%%%%%
 \\ Now let us consider $(1)_{L,R}$. First we have $(1)_{L}=0$.
Let us compute $(1)_{R}$ which is a term supported by the the mod $xg_0\otimes\psi_{1,{\mathbf n}_0}$.
\bea\label{1}
(1)_R &=& -\left(\frac{\eta_t}{1+\vert\eta_t\vert^2}(\alpha_3(t) +i\beta_3(t)) 
+\frac{1}{1+\vert\eta_t\vert^2}(\alpha_+(t) +i\beta_+(t))\right. \nonumber\\
 &+& \left.\frac{1}{1+\vert\eta_t\vert^2}(\alpha_-(t) +i\beta_-(t))\right)\cdot xg_0\otimes\psi_{1,{\mathbf n}_0}.\nonumber
\eea
Then we get that $(1)_R\neq 0$ if $\nabla_X\mathfrak{C}(t)\neq 0$.\\

\begin{remark} 
The section 4.2 of \cite{BG} use in a fondamental way their Lemma 4.5 concerning the action of the spin operators on spin coherent states. But the conclusion of this Lemma is false
 as we have shown in \eqref{schs}. In particular the $\sqrt{\spin}$ term  in our Lemma \ref{lem:sqs} is at  the origin of the wrong  ansatz \eqref{ans} in the regime $\hbar\spin=\kappa$.
\end{remark}
\section{More on the Dicke model}
\subsection{Preliminary computations and reductions}
A simpler form of the  interaction Hamiltonian for the Dicke  model is
\beq\label{fulldic}
\hat H_{\rm Dint}(t) = ((\cos t)x + (\sin t)\hbar D_x )\hbar S_1.
\eeq
More generally consider the Hamiltonian $\hat H(t) = (\alpha(t)x + \beta(t) D_x ))A$, where $A$ is a bounded Hermitian operator in the Hilbert space ${\mathcal H}$.  
 Let be ${\mathcal U}(t)$  the propagator for the  Schr\"odinger equation with initial  data at $t=0$.
$$
i\partial_t\Psi(t) = \hat H(t)\Psi(t)
$$
Then we have the following lemma related with Campbell-Hausdorff formula.
\begin{lem}
There exist two  scalar functions $c(t)$ and $\tilde c(t)$ such that
\beq\label{CH1}
{\mathcal U}(t) ={\rm e}^{-\frac{i}{2}c(t)A^2}{\rm e}^{-ib(t)D_xA}{\rm e}^{-ia(t)xA}
\eeq
where $c(t) = \int_0^t\alpha(\tau)b(\tau)d\tau$, $a(t) = \int_0^t\alpha(\tau)d\tau$, $b(t) = \int_0^t\beta(\tau)d\tau$, and
\beq\label{CH2}
{\mathcal U}(t) ={\rm e}^{-\frac{i}{2}\tilde c(t)A^2}{\rm e}^{-i(b(t)D_x  + a(t)x)A}
\eeq
where $\tilde c(t) = c(t)-a(t)b(t)$.
\end{lem}
\begin{proof}  Let us first remark that \eqref{CH2} is a direct consequence  of the Campbell-Hausdorff formula. So it is  enough to prove \eqref{CH1}.
Let $V(t) = {\rm e}^{-ib(t)D_xA}{\rm e}^{-ia(t)xA} $. By a direct computation  and a commutation we get that that 
$$
i\partial_tV(t)= (\dot a(t) xA +\dot b(t)D_x)V(t) +  a(t)b(t)A^2.
$$
So we get \eqref{CH1} with $\dot c(t) = \alpha(t)b(t)$. Now from the Campbell-Hausdorff formula, we get \eqref{CH2} with $\tilde c(t) = c(t) -a(t)b(t)$.
\end{proof}

$\square$.
%For the Dicke model it is possible to compute the spin-orbit  interaction using a known formula  for  exponential  of spin matrices in arbitrary spin number $\spin$ \cite{CFZ}.\\

Let us assume that $\Psi(0,x) = \varphi_{z_0}(x)\psi_{n_0}$, the product of a Schr\"odinger coherent state  and a spin coherent   state. Let us  assume first that $\beta=0$. Then   the time evolution is given by
 $\Psi(t,x) = \varphi_{z_0}(x){\rm e}^{-ia(t)xS_1}\psi_{\bf n_0}$. We have seen (Proposition \ref{tind}) that
$$
{\rm e}^{-ia(t)xS_1}\psi_{\bf n_0} = {\rm e}^{i{\spin}\gamma(t,x)}\psi_{{\bf n}(t,x)}.
$$
The classical spin ${\bf n}(t,x)$ is given here by 
\bea
n_1(t,x)&=&n_{0,1}=\cos\theta_0,\nonumber\\
n_2(t,x)&=& \sin\theta_0\cos(\theta(t,x)),\nonumber\\
n_3(t,x)&=&\sin\theta_0\sin(\theta(t,x)),\nonumber
\eea
where $\theta(t,x) = \theta_0 -a(t)x$.
The phase $\alpha(t,x)$ is given by
\beq\label{alpha}
\dot\alpha(t,x) =\frac{\alpha(t)x}{1+\vert\eta\vert^2}\left(\frac{n_1(t,x)}{1+n_3(t,x)} +\vert\eta\vert^2\frac{n_2(t,x)}{1+n_3(t,x)}\right)
\eeq
 with  $\gamma(t_0,x)=0$ and  $\vert\eta\vert^2 = \frac{1-n_3^2}{(1+n_3)^2}$.
% \end{document}
%Notice that in applying formula \eqref{action} here $H_c$ depends linearly in $x$ so that  $\alpha(t,x) = x\beta(t)$ hence we have strong quantum oscillations for any localisation 
%of the field in $\R\backslash 0$  when 
% $\spin=\frac{\kappa}{\hbar}$  goes to $+\infty$.\\
  %It could be an explanation for the failure of the semi-classical ansatz \eqref{ans} in this regime. More investigations have to be done.
 So we have the formula
 
 $$
 \Psi(t,x) = \varphi_{z_0}(x) {\rm e}^{i{\spin}\gamma(t,x)}\psi_{{\bf n}(t,x)}.
 $$
 If $\spin$  is  frozen ($\hbar$-independent )  we can use the Taylor formula and we get
 $$
 \Vert\Psi(t) -\varphi_{z_0} {\rm e}^{i{\spin}\gamma(t,q_0)}\psi_{{\bf n}(t,q_0)}\Vert_{L^2(\R, {\mathcal H}_{\spin})} = O(\sqrt\hbar)
 $$
 which is a compatible with  the ansatz \eqref{ans} for $\mu=1/2$.

 For  $\spin\hbar =\kappa>0$ we also consider Taylor expansion. 
 $$
 \gamma(t, x) = \gamma(t,q_0) + \partial_x\gamma(t,q_0)(x-q_0) +  \partial_x^2\gamma(t,q_0)\frac{(x-q_0)^2}{2} + O(\vert x-q_0\vert^3)
 $$
% \end{document}
 
 Hence
 $$
 \varphi_{z_0}(x){\rm e}^{i{\spin}\gamma(t,x)} = \varphi_{z_0}(x){\rm e}^{i{\spin}(\gamma(t,q_0) + \partial_x\gamma(t,q_0)(x-q_0) +  \partial_x^2\gamma(t,q_0)\frac{(x-q_0)^2}{2} )} + O(\kappa \sqrt\hbar)
 $$
in this formula we see that the spin motion add a momentum to the orbit motion (the linear term in $(x-q_0)$ and a   quadratic contribution  to the Gaussian.\\
But the Taylor argument  to eliminate the $x$-dependence does not work for the spin part  $\psi_{{\bf n}(t,x)}$.\\
The reason is the following. From explicit formulas  \cite{CR}, we know that 
\beq\label{overlap}
\vert\langle\psi_{\bf n}, \psi_{\bf m}\rangle\vert = {\rm e}^{\spin\log(1-\frac{\vert{\bf n}-{\bf m}\vert^2}{4})},
\eeq
hence we have 
$$
\Vert\psi_{\bf n}-\psi_{\bf m}\Vert^2 \approx 2(1-{\rm e}^{-s\frac{\vert \bf n-\bf m\vert^2}{4}})
$$
With the localization by $\varphi_{z_0} $  for $x$ in a neighborhood of $q_0$ we have $\vert\bf n(t,x)-\bf n(t,q_0)\vert$  of order $\hbar$.
Assume that $0<\theta_0< \pi/2$ and $f=1$. Then there exist $c_1>,  c_2>0$ such that for $t>0$ we have
$$
\Vert\psi_{\bf n(t,x)} -\psi_{\bf n(t,q_0)}\Vert^2 \geq c_1 (1-{\rm e}^{-c_2\spin t^2\vert x-q_0\vert^2}).
$$
Using that  $\spin =\frac{\kappa}{\hbar}$, we get with $c_3>0$, for $t_0>0$ small enough, 

$$
\Vert\varphi_{z_0}(x)(\psi_{\bf n(t,q_0)}-\psi_{{\bf n}(t,x)})\Vert^2_{L^2(\R_x,{\mathcal H}_{2\spin+1})} \geq c_3\kappa t^2,
\;\forall t\in]0, t_0].
$$
This  shows that the classical spin ${\bf n}(t,x)$ has to depend not only on the orbit but also on the position $x$ on the orbit which is not compatible with the ansatz \eqref{ans}.

Another and more accurate way  to understand the last computations is related with an entanglement-decoherence   phenomenon  for the time evolution of the initial state  $\Psi(0)= \varphi_{z_0}\otimes\psi_{\bf n_0}$,  when the interaction with  a large spin system  is switched on. \\
We shall see that  for $t\in[0,t_0], \kappa>0$ and $\spin=\frac{\kappa}{\hbar}$, then $\Psi(t,x)$ is an {\em entangled state} which  means  that it is not possible to have a decomposition like  $\Psi(t) =\varphi(t)\otimes\psi(t)$ with  $\varphi(t)\in L^2(\R)$ and $\psi(t)\in{\mathcal H}_{2\spin+1}$.\\
%cannot be a tensorial product of an orbit state times a spin state at the leading order in $\hbar$ if $\spin =\frac{\kappa}{\hbar}$   with $\kappa>0$. This is reminiscent of a decoherence phenomenum.
%%%%%%%%%%%%%%%%%%%%%%%%%%%%%%%%%%%%%
%%%%%%%%%%%%%%%%%%%%%%%%%%%%%%%%%%%%%%%
We use the partial  traces   in the Hilbert space $L^2(\R,\cc)\otimes{\mathcal H}_{2\spin+1}= L^2(\R, {\mathcal H}_{2\spin+1}) $.\\
Recall that in a tensor product  of Hilbert spaces ${\mathcal H}= {\mathcal H}_1\otimes{\mathcal H}_2$,  for a trace class operator $A$ in ${\mathcal H}$,  the partial trace of $A$ 
on ${\mathcal H}_2$ is the unique trace class operator in  ${\mathcal H}_1$, denoted ${\rm tr}_{{\mathcal H}_2}(A)$,  such that the  following Fubini identity is satisfied for any bounded operator $B$ on ${\mathcal H}_1$, 
$$
{\rm tr}_{{\mathcal H}_1}({\rm tr}_{{\mathcal H}_2}(A)B) = {\rm tr}_{\mathcal H}(A(B\otimes{\mathbb I}_{{\mathcal H}_2})).
$$
We shall use the following invariance property: if $U_k$ are invertible operators in ${\mathcal H}_k$, and  $U=U_1\otimes U_2$,   then we have 
$$
{\rm tr}_{{\mathcal H}_2}(U^{-1}AU) = U_1^{-1}{\rm tr}_{{\mathcal H}_2}(A)U_1.
$$
 
Notice that if ${{\mathcal H}_2}=\cc$ then we have  ${\mathcal H}_1={\mathcal H}_1\otimes\cc$ and ${\rm tr}_{\cc}(A)=A$.\\
Let $\Psi =\psi_1\otimes\psi_2\in{\mathcal H}_1\otimes{\mathcal H}_2$. and denote by $\Pi_{\Psi}$ the orthogonal projector on $\Psi$. Then we have 
${\rm tr}_{{\mathcal H}_2}\Pi_{\Psi} = \Pi_{\psi_1}$.  \\
{\em Physical interpretation}: if $A$ is a density matrix in ${\mathcal H}$ (non negative operator with trace 1) then 
${\rm tr}_{{\mathcal H}_2}(A)$ is also a density matrix in ${\mathcal H}_1$ which represents  the state of the sub-system  ${\mathcal H}_1$.\\
Suppose that the total system satisfies a Schr\"odinger equation with an initial state $\Psi_0=\psi_1\otimes\psi_2$ (pure state in ${\mathcal H}$).
For  time $t>0$  the density matrix $\Pi_{\Psi(t)}$  of the   total system  is pure  
but the density matrix  ${\rm tr}_{{\mathcal H}_2}(\Pi_{\Psi(t)})$ of the subsystem in ${\mathcal H}_1$  is not necessary  a pure state because it is not isolated (decoherence ).
When the rank of  ${\rm tr}_{{\mathcal H}_2}(\Pi_{\Psi(t)})$ is $\geq 2$  the sub-system (1) can occupy at least two orthogonal pure states with probabilities in $]0, 1[$, 
 like for the Schr\"odinger cat.

%%%%%%%%%%%%%%%%%%%%%%%%%%%%%%%%%%%%%%%%%%%
%%%%%%%%%%%%%%%%%%%%%%%%%%%%
Let  be $\Psi \in L^2(\R,\cc)\otimes {\mathcal H}_{2\spin +1})\simeq L^2(\R, {\mathcal H}_{2\spin +1})$.
A simple computation gives
$$
{\rm tr}_{{\mathcal H}_{2\spin +1}}(\Pi_\Psi)f(x) = \int_\R f(y)\langle\Psi(y), \Psi(x)\rangle_{{\mathcal H}_{2\spin +1}}dy,\;\; \forall f\in L^2(\R).
$$
In the same way we have also
$$
{\rm tr}_{L^2(\R)}u = \int_\R\langle\Psi(x),u\rangle_{{\mathcal H}_{2\spin +1}}\Psi(x)dx,\;\; \forall u\in {{\mathcal H}_{2\spin +1}}.
$$
The orbital density matrix $\rho_{\rm O}(t):= {\rm tr}_{{\mathcal H}_{2\spin +1}}(\Pi_{\Psi(t)})$ has the integral kernel 
$K(t,x,y) = {\rm e}^{i\spin(\gamma(t,x)-\gamma(t,y)}\tilde K(t,x,y)$ where  
\beq\label{kern}
\tilde K(t,x,y)= \overline{\varphi_{z_t}(x)}\varphi_{z_t}(y)\langle\psi_{{\bf n}(t,x)},\psi_{{\bf n}(t,y)}\rangle_{{\mathcal H}_{2\spin +1}}.
\eeq
   %Let us consider the operator $\hat K(t)$ in $L^2(\R,\cc)$ with the  kernel $ \tilde K(t,x,y)$.
  So we have
  $$
  {\rm tr}(\rho_O(t)^2) = \int_{\R^2}\vert \tilde K(t,x,y)\vert^2dxdy.
  $$
  Recall that  $\hat H(t) = \alpha(t)xS_1$.
\begin{lem}\label{estim:part}
There exists $c_0>0$ such that for any $0<\mu<1/2$ we have
\beq\label{decoh}
\int_{\R^2}\vert\tilde K(t,x,y)\vert^2dxdy \leq (1+c_0\kappa t^2)^{-1/2} + O(\hbar^\mu)
\eeq
\end{lem}
Notice that $\hat K(t)$ is a non negative operator of trace 1. So let be $\lambda_j(t)$ the eigenvalues  of $\rho_O(t)$ (in decreasing order with multiplicities). So if 
 $t>0$, $\kappa>0$ and $\hbar$ small enough,  we get from the Lemma  that $$\sum_{j}\lambda_j^2(t) < \sum_{j}\lambda_j(t)=1$$  hence $\hat K(t)$ has at least two eigenvalues $<1$ 
when the spin interaction is switched on (recall that $\hat K(0)$ is the projector on a pure state).

{\em Proof} of Lemma \ref{estim:part}.\\
We use \eqref{overlap} to compute the Schwartz kernel \eqref{kern} and Lemma \ref{deltan} \\
So we get for any $\mu>0$, 
$$  
\int_{\R^2}\vert\tilde K(t,x,y)\vert^2dxdy \leq 
(\pi\hbar)^{-1} \int_{\{u^2+v^2\leq r_0\}}{\rm e}^{-\frac{1}{\hbar}(u^2+v^2 +c_1\kappa t^2(u-v)^2)}dudv
$$
Then a direct computation gives the estimate:
$$
\int_{\R^2}\vert\tilde K(t,x,y)\vert^2dxdy \leq (1+2c_1\kappa t^2)^{-1/2} + O(\hbar^\mu)
$$
$\square$
%%%%%%%%%%%%%
%\textcolor{red}{ Until now we dont know if an analogous property is true for the full Dicke Hamiltonien \eqref{fulldic}. Under study!}

Let us now consider the full interaction Hamiltonian for the Dicke model: $\hat H_{\rm Dint}(t) = ((\cos t)x + (\sin t)\hbar D_x )\hbar S_1$.
By a symplectic transform this Hamiltonian is conjugate to the previous one. So Lemma \ref{estim:part} is also true in this case.
\subsection{Proof of Theorem \ref{Dmod}}
  Strategy:\\
1) Reduce to the interaction picture with the propagator ${\mathcal V}(t,t_0)$:\\
Denote $\Psi_I(t)= {\mathcal V}(t,t_0)\Psi_0$. Then we have $\Pi_{\Psi(t)} ={\mathcal U}_0(t-t_0)\Pi_{\Psi_I(t)}{\mathcal U}_0(t_0-t)$.
 Hence we have 
 $$
 {\rm tr}_{{\mathcal H}_{2s+1}}\Pi_{\Psi(t)}=  {\mathcal U}_0(t-t_0){\rm tr}_{{\mathcal H}_{2s+1}}\Pi_{\Psi_I(t)}{\mathcal U}_0(t_0-t)
$$
But ${\mathcal U}_0(t-t_0)$ is a unitary operator in $L^2(\R,\cc)$ so it is enough to prove the result for $\Psi_I(t)$.\\
2) Proof of the result if $\omega_3=0$\\
Let us consider the time dependent Hamiltonian $\hat H(t) = (\alpha(t)x + \beta(t) D_x ))S_1$. Using \eqref{CH2},  to compute partial trace on 
${\mathcal H}_{2\spin +1}$ it is enough the consider the propagator:
$$
W(t):={\rm e}^{-i(b(t)\hbar D_x  + a(t)x)S_1}
$$
Finally by a symplectic rotation $R(t)$ we get a metaplectic transformation $\hat R(t)$ in $L^2(\R,\cc)$  such that 
$$
\hat R(t)W(t)\hat R^*(t) = {\rm e}^{\tilde\alpha(t)xS_1}: ={\tilde W}(t).
$$
We have already proved above the decoherence for the evolution of the initial state $\Psi(0,x) = \varphi_{z_0}(x)\psi_{n_0}$ by ${\tilde W}(t)$. So the proof for $\omega_3=0$ is achieved.\\
3) The case $\omega_3\neq 0$.\\
Now the interaction is computed from the decoupled Hamiltonian \\$\hat K_0= \hat H_0 + \hbar\omega_3 S_3$. So we have
$\cmU(t) = \cmU_0(t){\mathcal V}(t)$, where 
$ \cmU_0(t) = {\rm e}^{-i\frac{t}{\hbar}\hat K_0}$ and 
the propagator ${\mathcal V}(t)$ must satisfy
$$
i\hbar\partial_t{\mathcal V}(t) =\hat K_I(t){\mathcal V}(t),\;\;{\mathcal V}(0)=\mathbb I,
$$
where
 $$
 \hat K_I(t) = \left(\alpha(t)x + \beta(t)\hbar D_x\right)\left(\cos(\omega_3 t)\hat S_1 + \sin(\omega_3t)\hat S_2\right).
 $$
Like in the case $\omega_3=0$, to compute the partial trace it is enough to consider the case $\beta(t)=0$.
 Hence we can conclude, like for $\omega_3=0$,  using Proposition \ref{tind},   \eqref{overlap} and the following  Lemma
 \begin{lem}\label{deltan} Let us consider the Landau equation depending on the parameter $x\in\R$,
 $$
  \partial_t{\bf n}(t) =\mathfrak{C}(t,x)\wedge{\bf n}(t).
  $$
  Assume that $\mathfrak{C}_3\equiv 0$, $\nabla_x\mathfrak{C}_1(t,x_0)\neq 0$ and $n_3(0)\neq 0$. Then 
 there  exists $c_0>0$, $t_0>0$, $r_0>0$  such that 
\beq\label{nest}
 \vert{\bf n}(t,x)-{\bf n}(t, y)\vert \geq c_0t\vert x-y\vert,\;\; {\rm if}\;\; \vert x_0-y\vert+ \vert x_0-x\vert\leq r_0, 0\leq t\leq t_0.
 \eeq
 \end{lem}
\begin{proof} From  the   Landau-Lifshitz equation we get 
$\vert\partial_s\partial_x n_2(s,u)\vert\geq c_1>0$ from $s$ and $u-x_0$ small enough. Then \eqref{nest} follows.
\end{proof}
%%%%%%%%%%%%%%%%%%%%%%%%%
%%%%%%%%%%%%%%%%
\begin{remark}
It  seems possible  that these results coud be extended to more general spin-orbit interaction like $\widehat{\mathfrak C}_1S_1$ such that 
 $\nabla_X{\mathfrak C}_1(X)\neq 0$ (principal type condition) by constructing  a unitary Fourier-integral operator $U$ such that
 $U\widehat{\mathfrak C}_1U^* \approx \hat x$.
 \end{remark}
 \appendix\section{Classical perturbations of Hamiltonians}
 Our aim here is to analyze the consequence on the trajectories of the perturbation of order $\kappa$ in \eqref{cl:spin0}
  for the critical regime $\spin \approx \hbar^{-1/2}$. More generally let us consider times dependent smooth vector fields  in 
  $\R^m$, $F(t, X),
   G(t, X)$ and the differential equations
   $$ 
   \dot Y(t) = F(t,Y(t)),\; \dot X(t) = F(t, X(t)) +\kappa G(t, X(t)),\; X(0=Y(0)=X_0.
   $$
   \begin{lem}\label{pertkappa} There exists $t_0>0, \kappa_0>0$ such that for  we have 
   $$
   \Vert X(t) -Y(t)\Vert = \kappa t\Vert G(0,X_0)\Vert + O(\kappa t^2),\; {\rm for}\;\; 0<t<t_0, 0<\kappa<\kappa_0.
   $$
  In particular if $G(0,X_0) \neq 0$ we have, for $c_2>0$, 
  $$
  \Vert X(t) -Y(t)\Vert \geq c_2\kappa t, \;\;{\rm for}\;\; 0<t<t_0.
  $$
   \end{lem}
   \begin{proof}
   The following argument is standard to compare solutions of differential equations.\\
   In a first step, we get for some $c_1>0$, $t_0>0$,   that we have
   $$
    \Vert X(t) -Y(t)\Vert \leq c_1\kappa, \;\;{\rm for}\;\; 0<t<t_0.
    $$
   Then  we use the variation equation with the linearized equation around $Y(t)$ to get
   $$
   X(t) -Y(t) = \kappa\int_0^t R(s,t)G(s,Y(s))ds + O(\kappa^2 t^2).
   $$
   So we get the Lemma with $t_0>0$ small enough.
   \end{proof}
 \section{Proof  of Proposition \ref{tind}}\label{pftind}
 Here we can assume $\hbar=1$.\\
 We are using complex coordinates on the sphere for the Hamiltonian (Section \ref{clspin}):
 $$
H_c(\eta, \bar\eta) = (1+\vert\eta\vert^2)^{-1}(\mathfrak C_3(1-\vert\eta\vert^2) -(\mathfrak C_-\bar\eta +\mathfrak C_+\eta))
$$
Recall that $\eta = \frac{in_2 -n_1}{1+n_3}$  for $\bf n =(n_1,n_2,n_3)$ ($\eta(0,0,-1) =\infty$)\\
Let us compute the time derivative of the ansatz \eqref{no},  like in Section \ref{crit}. The classical dynamics of the spin is given  by \eqref{clandau}.\\
By a straightforward computation we get the spin phase $\alpha$:
$$
\alpha(t)  = \int_0^t\left( \frac{\Im(\dot\eta\bar\eta)}{2(1+\vert\eta\vert^2)}-H_c(\eta, \bar\eta)\right)d\tau =\gamma(t) .
$$

%A Compact Formula for Rotations
%as Spin Matrix Polynomials
%Thomas L. CURTRIGHT †, David B. FAIRLIE ‡ and Cosmas K. ZACHOS §

\end{document}